\documentclass{article}

\usepackage[utf8]{inputenc}		
\usepackage{authblk}

\usepackage{amsmath}
\usepackage{amsthm}
\usepackage{amsxtra}
\usepackage{mathpazo}
\usepackage{MnSymbol}
\usepackage{marvosym}

\usepackage{enumerate}
\usepackage{braket}
\usepackage{url}
\usepackage{hyperref}

\DeclareMathAlphabet{\mathsc}{OT1}{cmr}{m}{sc}
\allowdisplaybreaks

\newcommand{\lo}[1]{\raisebox{-0.1ex}{$#1$}\,}
\newcommand{\loo}[1]{\raisebox{-0.2ex}{$#1$}\,}
\newcommand{\Lo}[1]{\raisebox{-0.3ex}{$#1$}\,}

\newcommand{\LO}[1]{\raisebox{-0.5ex}{$#1$}\,}

\newcommand{\R}{\mathbb R}

\newcommand{\Z}{\mathbb Z}

\newcommand{\T}{\mathbb T}

\newcommand{\eps}{\varepsilon}
\newcommand{\abs}[1]{\lvert #1 \rvert}

\newcommand{\norm}[1]{\lVert #1 \rVert}

\newcommand{\EE}[3]{\mathbb E_{#1}\!\left[#2\,\middle|\,#3\right]}
\newcommand{\D}[1]{\mathrm{d}#1}
\newcommand{\I}{\mathrm{i}}
\newcommand{\e}{\mathrm{e}}
\newcommand{\Langle}{\left\langle}
\newcommand{\Rangle}{\right\rangle}
\newcommand{\Cb}[1]{C_\mathrm{b}\!\left(#1\right)}

\newcommand{\mc}[1]{\mathcal{#1}}

\theoremstyle{definition}
\newtheorem{defn}{Definition}

\theoremstyle{plain}
\newtheorem{thm}[defn]{Theorem}

\theoremstyle{remark}
\newtheorem{ex}[defn]{Example}

\title{Towards an Analytic Theory of Stochastic and Quantum Fields}
\author{
Rodrigo Vargas Le-Bert\footnote{\Letter {\tt vargonis@gmail.com}. } 
}
\date{\today}

\begin{document}

\begin{titlepage}
\maketitle
\thispagestyle{empty}

\begin{abstract}
We 
propose a method for
the rigorous construction of physically relevant functional measures. In shaping it we get several conceptual insights, which can perhaps be summarized by the following statement: the renormalized interaction Lagrangian should be the generator of a flow on a space of asymptotically free cylinder functional measures with density given, in the case of Boson fields with polynomial self-interaction, by a generalized form of the Appell polynomials. 
\end{abstract}

\setcounter{tocdepth}{2}
\tableofcontents
\end{titlepage}

\section{Introduction}

In this work we propose a general method for the construction of quantum fields. 
We see the problem 
as naturally divided in two broad stages:
\begin{enumerate}
	\item
Let $X$ be a space of fields $x=x(t;s_1\lo,\dots,s_{d-1})$. Quotienting out the high-energy modes gives a finite-dimensional projection $P:X\rightarrow X_P$\lo. Consider the family $\mc P$ of all such projections; a \emph{cylinder measure} is a compatible family $\mu = \Set{\mu_P | P\in\mc P}$ of measures on their ranges---or, equivalently, a positive linear functional on the algebra $\Cb{X;\mc P}$ of continuous, bounded, \emph{cylinder functions}, i.e.\ those $f:X\rightarrow\R$ which factor through one of the projections $P\in\mc P$. The first stage is the construction of the cylinder measure corresponding to the Euclidean path integral of the theory.
	\item
Cylinder measures are, generally speaking, only finitely additive. The second stage is to Radonify the cylinder measure on a space $\overline X\supseteq X$ whose elements are such that $t\mapsto x(t;\cdot)$ is continuous, thus obtaining a stochastic process which, modulo OS-positivity, is suitable for the reconstruction of the Lorentzian field.
\end{enumerate}
The second stage can be approached using Radonifying operators and an\-i\-so\-tro\-pic Sobolev spaces. The first one presents more conceptual difficulties, for it is there that renormalization takes place. Our approach stems from a rigorous consideration of the deformation of the free measure which should occur upon gradually turning the interaction on. 
Thus, we take a hypothetic family of measures $\set{\mu_\lambda}$, where $\lambda\geq 0$ stands for a coupling constant, and let the formal expression
\[
\mu_\lambda = \e^{-\lambda L}\mu_0\loo,
\]
where $L$ is a renormalized interaction Lagrangian, guide our developments. Now, we take this equation to mean that for small $\eps$, in some suitable sense 
\[
\mu_{\lambda+\eps} \approx \mu_\lambda + \eps L\mu_\lambda\loo.
\]
However, $\mu_\lambda$ is not just one measure, but a whole family $\set{\mu_{\lambda,P}}$ of them, and therefore $L$ itself must stand for a family $\set{L_P}$ of functions $L_P=L_P(x_P)$, where $x_P=P(x)$, such that
\begin{enumerate} 
	\item $\Set{L_P(x_P)\mu_{\lambda,P}(\D x_P)}$ is compatible.
	\item $L_P$ is \emph{local,} or at least approximately so, in the sense that
\[
L_P(x_P) \approx \int \mc L_P\bigl(x_P(s)\bigr)\D s,\quad s=(t;s_1\lo,\dots,s_{d-1}),
\]
for some effective Lagrangian density $\mc L_P:\R\rightarrow\R$---assuming, for simplicity, that we are dealing with a real scalar Boson field with a self-interaction that does not depend on its derivatives.
\end{enumerate}
The first condition actually forces $L_P$ to depend on $\mu_\lambda$\lo, and it is therefore better to write
\[
L(\mu) = \Set{ L_P(\mu;x_P)\mu_P(\D x_P)}.
\]
Having gained this structural understanding of the problem of constructive field theory, our task is clear: to construct the non-linear operator $L(\mu)$, which must somehow be related to the classical interaction Lagrangian density, and then show that it generates a flow on a space of cylinder measures. 

Before proceeding, a word of caution regarding terminology is in order. Our considerations have led us to think of the function $L_P=L_P(\mu;x_P)$ above as the effective Lagrangian, whereas traditionally the effective Lagrangian is the function $\tilde L_P = \tilde L_P(\lambda;x_P)$ such that
\[
\mu_{\lambda,P} = \e^{-\tilde L_P(\lambda)}\mu_{0,P}\lo.
\]
We will call this function the \emph{aggregate} effective Lagrangian. Its relationship with our effective Lagrangian $L_P$ can be theoretically worked out from the fact that $\mu_{\lambda,P}$ solves the differential equation
\begin{equation}\label{measure evolution}
\frac{\D}{\D\lambda}\mu_P = -L_P(\mu)\mu_P\lo,
\end{equation}
the problem being that $L_P$ is a function of the whole family $\mu=\set{\mu_P|P\in\mc P}$, as opposed to just $\mu_P$\lo.

\section{Effective theories and cylinder measures}

Cylinder measures are traditionally defined as compatible collections of measures on the set of all finite dimensional quotients of a locally convex vector space---see~\cite{badrikian1969mesures}, for instance. This is insatisfactory for two reasons: first, it innecessarily leaves out non-linear spaces, hiding the fact that the notion is just at the coordinate system level; second, but more relevant to this work, in asking for a measure on \emph{every} finite dimensional quotient it becomes impractical, leaving us with little more than one example: Gaussian measures. Thus, we adopt a different definition, one that makes it evident its close relationship with renormalization theory.

\subsection{Cylinder measures}

Let $\mc P$ be a directed set, $\set{X_P|P\in\mc P}$ a projective system of topological spaces with projective limit $\overline X$ and canonical projections $\pi_P:\overline X\rightarrow X_P$\loo, and $X$ a subspace of $\overline X$ which is \emph{full,} in the sense that $\pi_P(X)=X_P$\loo. By a harmless abuse of notation, we will usually write $P$ instead of $\pi_P$\loo.
It will also be convenient to write $P$ for the projection $X\rightarrow X_P$\loo, and even for the projection $X_Q\rightarrow X_P$ when $Q\succcurlyeq P$ is understood from the context. 
We think of $\set{P:X\rightarrow X_P|P\in\mc P}$ as a coordinate system on $X$. 
Now, given $P\in\mc P$, consider the algebra
\[
\Cb{X;P} = \Set{\vphantom{\hat A} f\in \Cb{X} | f \text{ factors through } P:X\rightarrow X_P }.
\]
If $P\preccurlyeq Q$, there is a natural inclusion $
\Cb{X;P}\hookrightarrow \Cb{X;Q}$. The resulting directed system has an algebraic injective limit 
\[
\Cb{X;\mc P} = \injlim \Set{\vphantom{\hat A} \Cb{X;P} | P\in\mc P }.
\]
The elements of $\Cb{X;\mc P}$ are called \emph{cylinder functions} (for the coordinate system in use). 
Some important examples of coordinate systems follow.

\begin{ex}
Let $X$ be a Tychonoff space and $\Set{e_i|i\in\mc I}\subseteq \Cb{X}$ a separating family of continuous functions. Given a finite subset $I\subseteq\mc I$, consider the equivalence relation
\[
x\sim y \Leftrightarrow (\forall i\in I)\ e_i(x)=e_i(y)
\]
and let $P_I:X\rightarrow X_I$ be the corresponding quotient. One has that $X$ is a full subspace of $\projlim X_I$\loo. 
Now, for $I=\{i_1\lo,\dots,i_n\}$, we define 
\[
e_I:X\rightarrow\R^n,\quad e_I(x) = \left(e_{i_1}(x),\dots e_{i_n}(x)\right)
\]
and then $f\in\Cb{X}$ is a cylinder function if, and only if, it factors through one of the $e_I$'s.
\end{ex}

\begin{ex}
Consider a path space $X_I=C(I,X)$, with $I\subseteq\R$. Given a finite number of time instants $t_1\lo,\dots,t_n\in I$, we have the projection
\[
x\in X_I\mapsto \left( x_{t_1}\Lo,\dots, x_{t_n} \right) \in X^n.
\]
The resulting coordinate system on $X_I$ has cylinder functions 
\[
x\mapsto f\left( x_{t_1}\Lo,\dots x_{t_n} \right),\quad f:X^n\rightarrow\R.
\]
The classical Kolmogorov consistency theorem is about the construction of path space measures in this coordinate system.
\end{ex}

\begin{ex}
Let $X$ be a Banach space and $\mc P\subseteq B(X)$ a directed family of projections converging strongly to $1\in B(X)$ (so that $X$ has the metric approximation property). The importance of this convergence hypothesis will be shortly seen. This is the example that corresponds more closely to the situation studied in standard cylinder measure theory. A particular case is that of a separable Hilbert space with orthonormal basis $\set{e_i}$ and projections $P_n = \sum_{i\leq n}e_ie_i^*$\loo, where $e_i^*(x) = \langle e_i\lo,x\rangle$. Cylinder functions, then, are those which depend only on a finite number of coordinates.
\end{ex}

\begin{ex} \label{deltas as coordinates}
Let $(S,\D s)$ be a measure space. In this example we show how to rigorously treat the so-called physical coordinates in field theory: that is, the notion that
\[
\Set{\delta_s=\delta(\cdot-s)|s\in S},
\]
where $\delta$ is Dirac's delta, is a ``basis'' for the space of fields $x:S\rightarrow\R$. We take a system $\set{p_i | i=1\dots n}$ of projections of the von Neumann algebra $L^\infty(S)$ which is orthogonal and complete, in the sense that $p_ip_j=0$ and $\sum p_i=1$. 
To $\{p_i\}$ we associate the conditional expectation
\[
P:L^\infty(S)\rightarrow X_P\loo,\quad
P = \sum_i p_ip_i^*\loo,\ p^*(x) = \int_S \bar px,
\]
where $\bar p = p/\abs p$ and $\abs p = \int_X p$. 
Now, let $\set{q_{ij} | i=1\dots n,\,j=1\dots m}$ be a refinement of $\{p_i\}$, i.e.\ another complete system of orthogonal projections such that $p_i=\sum_{j} q_{ij}$, with associated conditional expectation $Q:L^\infty(S)\rightarrow X_Q$\loo. Since $\{q_{ij}\}$ is a refinement of $\{p_i\}$, we have a projection (conditional expectation) $X_Q\rightarrow X_P$\loo. Given a directed family $\mc P$ of such systems of orthogonal projections we get a projective system $\{X_P\}$ and, if the family generates $L^\infty(S)$, then any good $X\subseteq L^\infty(S)$ will become a full subspace of $\projlim X_P$\loo. 
\end{ex}

\begin{defn}
Let $X$ be a Tychonoff space equipped with a coordinate system $\Set{P:X\rightarrow X_P|P\in\mc P}$.
A \emph{cylinder measure} on $X$ is a family of Radon measures $\set{\mu_P\text{ on } X_P}$ which is \emph{compatible,} in the sense that
\[
P_*\mu_Q = \mu_P\lo,\quad\text{ for all } Q\succcurlyeq P.
\]
We can also adopt a dual point of view and define a cylinder measure as a compatible family of positive linear functionals $\set{\rho_P:\Cb{X_P}\rightarrow\R}$---or, in other words, a positive linear functional on the injective limit $\Cb{X;\mc P}$. When the measure $\mu$ is clear from the context, we will sometimes write 
\[
\langle f\rangle = \int_X f(x)\mu(\D x) = \rho_P(f_P) = \int_{X_P}f_P(x_P)\mu_P(\D x_P),\quad f\in\Cb{X;\mc P},
\]
where $P\in\mc P$ and $f_P\in\Cb{X_P}$ are such that $f=f_P\circ P$.
Adopting a physical terminology, we will sometimes refer to a single measure $\mu_P$ on $X_P$ as an \emph{effective theory.}
\end{defn}

\subsection{Conditional expectations}

Given a measure $\mu$ on $X$, $f\in L^\infty(X,\mu)$ and $\varphi:X\rightarrow Y$, we will use the notation
\(
\EE{\mu}{f}{\varphi}
\)
for the conditional expectation of $f$ with respect to $\mu$ given $\varphi$---that is, the unique 
element $\EE{\mu}{f}{\varphi}\in L^\infty(Y,\varphi_*\mu)$ such that
\[
\int_X f(x)g\bigl(\varphi(x)\bigr)\mu(\D x) = \int_Y\EE{\mu}{f}{\varphi}(y)g(y)\varphi_*\mu(\D y),\quad g\in L^1(Y,\varphi_*\mu).
\]
Otherwise said, $\varphi_*(f\mu) = \EE{\mu}{f}{\varphi}\varphi_*\mu$, which enables one to see that $\EE{\mu}{f}{\varphi}$ is the Radon-Nikodým derivative of $\varphi_*(f\mu)$ with respect to $\varphi_*\mu$. Now, suppose that $\mu = \set{\mu_P|P\in\mc P}$ is a cylinder measure and $f,\varphi$ are cylinder functions. We define $\EE{\mu}{f}{\varphi}$ to be $\EE{\mu_P}{f_P}{\varphi_P}$, where $P$ is any projection in $\mc P$ such that $f$ and $\varphi$ factor through $X\rightarrow X_P$ via $f_P:X_P\rightarrow\R$ and $\varphi_P:X_P\rightarrow Y$, respectively. The result is easily seen to be independent of $P$, using the following properties of conditional expectations:
\begin{enumerate}
	\item $\EE{\mu}{(\varphi^*f)g}{\varphi} = f\EE{\mu}{g}{\varphi}$.
	\item $\EE{\mu}{f}{\varphi\circ\psi} = \EE{\psi_*\mu}{\EE{\mu}{f}{\psi}}{\varphi}$.
\end{enumerate}
Finally, given a symmetry $\varphi$ of $\mu$, i.e.\ a compatible family of invertible maps $\varphi_P:X_P\rightarrow X_P$ such that 
\(
(\varphi^{-1})_*\mu = \Set{(\varphi^{-1}_P)_*\mu_P} = \mu,
\)
one has the covariance property
\[
\EE{}{f\circ\varphi}{P} = \EE{}{f}{P}\circ\varphi.
\]

\subsection{The free real Boson field}

We recall now how the free measure of a real Boson field on the torus $\T^d$ looks like, in momentum coordinates.
Given $x\in L^2(\T^d)$, we have the Fourier series expansion
\[
x(s) = \sum_{k\in\Z^d}\hat x_k\e^{\I ks},\quad \hat x_k = 
\frac{1}{(2\pi)^d}\int_{\T^d}\e^{-\I ks}x(s)\D s.
\]
For a real field one has $\hat x_k^* = \hat x_{-k} $\loo. Write $[k] = \set{k,-k}$ for the equivalence class of $k$ in $\Z^d/\pm 1$ and
\[
a_{[k]}(x) = \frac{\hat x_k+\hat x_{-k}}{\sqrt 2},\quad b_{[k]}(x) = \frac{\hat x_k-\hat x_{-k}}{\sqrt 2\I},
\]
so that
\(
\norm x^2 = \sum_{k\in\Z^d} \abs{\hat x_k}^2 = \sum_{[k]\in\Z^d/\pm 1}\bigl(a_{[k]}^2+b_{[k]}^2\bigr)
\)
and $\D a_{[k]}\D b_{[k]} = -\I\D \hat x_k^*\D\hat x_k$\loo.
One has that
\[
\Set{x\in L^2(\T^d) | a_{[\ell]}(x) = b_{[\ell]}(x) = 0 \text{ whenever } [\ell]\neq[k] }
\]
is an eigenspace of $-\Delta$ with eigenvalue $\abs k^2 = \sum k_i^2$\LO. Thus, the free measure reads, formally,
\begin{align*} 
\mu(\D x) &=
\prod_{[k]\in\Z/\pm 1} \frac{\abs k^2+m^2}{2\pi}\e^{-\frac12(\abs k^2+m^2)\left(a_{[k]}^2+b_{[k]}^2\right)}\D a_{[k]}\D b_{[k]} \\
&= \prod_{[k]\in\Z^d/\pm 1} \frac{\abs k^2+m^2}{2\pi\I}\e^{-\frac12(\abs k^2+m^2)\hat x_k^*\hat x_k}\D\hat x_k^*\D\hat x_k\loo,
\end{align*} 
where $m$ is the mass parameter.
This is straightforwardly made sense of as a cylinder measure, and we do so using the system of projections
\[
P_Ax(s) = \sum_{k\in A}\hat x_k\e^{\I ks},\quad A\subseteq\Z^d \text{ with } \abs A<\infty \text{ and } A=-A.
\]
The measure $\mu_A$ on $P_AX=X_A$ is the product measure
\[
\mu_A(\D x_A) = \prod_{[k]\in A/\pm 1} \frac{\abs k^2+m^2}{2\pi\I}\e^{-\frac12(\abs k^2+m^2)\hat x_k^*\hat x_k}\D\hat x_k^*\D\hat x_k\loo,\quad x_A=P_Ax.
\]
It will be useful to have the special names $P_n$ and $P'_n$ for the projections $P_{B_n}$ and $P_{B_n'}$\LO, respectively, where
\[
B_n' = \Set{k\in\Z^d | \max\, \abs{k_i}= n},\quad
B_n = \bigcup_{m=0}^n B_m'\Lo,
\]
and we will use the following notation:
\begin{align*}
X_n &= P_nX& x_n &= P_nx& \mu_n &= \mu_{B_n}\\
X_n' &= P_n'X& x'_n &= P'_nx& \mu'_n &= \mu_{S_n}
\end{align*}

\section{Cylinder measure perturbation and renormalization theory}

\subsection{Cylinder measure perturbations}

\begin{defn}
A \emph{perturbation} of a cylinder measure $\mu=\set{\mu_P|P\in\mc P}$ is a family of measurable functions $L=\set{L_P}$, $L_P=L_P(x_P)$, such that $\Set{ L_P(x_P)\mu_P(\D x_P) }$ is a (possibly signed) cylinder measure. 
\end{defn}

Observe that the compatibility condition for the family $\set{L_P\mu_P}$ reads
\begin{equation} \label{renormalization}
L_P = \EE{\mu}{L_Q}{P},\quad Q\succcurlyeq P.
\end{equation}
Thus, a cylinder measure perturbation is a martingale $\set{L_P}$ with respect to the family of random variables $x_P = P(x)$. This suggests that it is potentially helpful to consider $\set{x_P}$ as a sort of stochastic process: one in which, instead of time passing, one has \emph{resolution increasing} (both being forms of information increase). When adopting that point of view we will talk of $\mu$ as a \emph{resolution process.} We note, in passing, that the free measure is a resolution process with independent increments.

Given a perturbation $L$ of a cylinder measure $\mu$, 
one has that the family $\Set{\e^{-\lambda L_P}\mu_P}$ is close (up to corrections of order $\lambda^2$) to being a compatible family of measures. Thus, establishing which perturbation of the free measure should correspond to a classical interaction Lagrangian amounts to solving the renormalization problem to order one. 
Also, if one is given an \emph{aggregate} effective interaction Lagrangian $\tilde L_Q$ at scale $Q$, then equation~\eqref{renormalization} provides a first-order approximation to the corresponding aggregate effective Lagrangian at scale $P$. Equation~\eqref{renormalization} can be seen, therefore, as a linearization of the renormalization step of passing from scale $Q$ to scale $P$.

\subsection{Local perturbations and Wick ordering}

Now, we are interested in \emph{local} perturbations, i.e.\ those arising asymptotically from integration of a density:
\[
L_n(x_n) = \int_{\T^d} \mc L_n\bigl(x_n(s)\bigr)\D s+o(1) \text{ as } n\rightarrow\infty.
\]
Observe that if $L_{n+1}$ is \emph{exactly} local, in the sense that
\[
L_{n+1}(x_{n+1}) = \int_{\T^d} \mc L_{n+1}\bigl(x_{n+1}(s)\bigr)\D s,
\]
then there is a priori no reason to expect $L_n(x_n) = \EE{}{L_{n+1}}{x_n}$ to be exactly local; hence the importance of defining locality via an asymptotic condition. There exist, however, perturbations which are exactly local at all scales.
In order for a family $\set{\mc L_n}$ of Lagrangian densities to provide such a perturbation it suffices that 
\[ 
\EE{\mu}{\mc L_m\bigl(x_m(s)\bigr)}{x_n} = \mc L_n\bigl(x_n(s)\bigr).
\] 
This helps understanding the importance of Wick ordering, given, for measures more general than the Gaussian, by the Appell polynomials. Take, for simplicity, the $\phi^4$ theory. Let $\mu$ be the free measure and define the quantities
\[
\sigma_p(m) = \Langle x_m(s)^p\Rangle,
\]
which are independent of $s$ because  $\mu$ is $\T^d$-invariant. We consider the polynomials
\begin{equation} \label{Wick 4}
\mc L_n(x) = x^4 - 6\sigma_2(n) x^2 + 6\sigma_2(n)^2-\sigma_4(n).
\end{equation}
Since $\mu$ has independent increments, the family $\Set{\mc L_n\bigl(x_n(s)\bigr)}$ has the martingale property 
with respect to $\set{x_n}$---see~\cite{anshelevich2003appell}, for instance. Of course, this applies as well to polynomial interactions other than the quartic. Now, since  commutation relations with field operators are not messed up by Wick ordering (see~\cite{baez1992introduction}, for instance), the effective Lagrangians $\set{\mc L_n}$ above have the best claim to be the perturbation of the free measure corresponding to the $\phi^4$ theory.

We are in position to understand that for a scalar boson field with polynomial self-interaction, $L(\mu)$ should be a perturbation of $\mu$ with density somehow given by the Appell version of the interaction---the problem being that for processes with \emph{dependent} increments the martingale property fails. 
In the next section, we show that a suitable generalization of the Appell polynomials  exists under an asymptotic freedom condition on $\mu$,
but before getting into it let us explicitely examine one renormalization step---that is, write down the relationship between $\mc L_n$ and $\mc L_{n-1}$\lo. Take $\mc L_n$'s of the form $\sum_{p=0}^4 \alpha_p(n)x^p$, with $\alpha_4(n)=1$. One has that $\EE{}{\mc L_n}{x_n}$ is the sum of the following terms:
\begin{align}
&x_{n-1}(s)^4 \label{quartic}\\
&\EE{}{ \alpha_3(n)+x_n'(s) }{x_{n-1}}x_{n-1}(s)^3 \label{cubic}\\
&\EE{}{ \alpha_2(n) + 3\alpha_3(n)x_n'(s) +  6x_n'(s)^2}{x_{n-1}} x_{n-1}(s)^2 \label{quadratic}\\
&\EE{}{ \alpha_1(n) + 2\alpha_2(n)x_{n}'(s) + 3\alpha_3(n)x_{n}'(s)^2 + 4x_{n}'(s)^3 }{x_{n-1}}x_{n-1}(s) \label{linear}\\
&\mathbb E\left[\alpha_0(n) + \alpha_1(n)x_{n}'(s) + \alpha_2(n) x_{n}'(s)^2 + \alpha_3(n)x_{n}'(s)^3 + x_{n}'(s)^4 \middle| x_{n-1} \right] \label{constant}
\end{align}
Observe the diagonal structure of this equations.
In the case of the free measure, \eqref{cubic} and~\eqref{linear} can be set to 0 for all $n$, but we have kept their general form here with an eye on the future.
With this expression at hand we emphasize two facts: 
\begin{enumerate}
	\item The Appell polynomials are by no means the only family of quartic polynomials providing a local perturbation of the free measure: given arbitrary $\alpha_p$'s at a scale $n_0$\lo, one can solve inductively for $\alpha_p(n)$. 
	\item If the resolution process $\mu$ has dependent increments, 
then $\alpha_p(n-1)$ might not be constant, even if $\alpha_p(n)$ is; it will generally be a non-local, non-homogeneous function $\alpha_p(n-1;x_{n-1};s)$. Thus, we must be prepared to accept that $\alpha_p(n) = \alpha_p(n;x_n;s)$. We note, however, that as the $s$-de\-pen\-dence comes from a conditional expectation, we can assume that the $\alpha_p$'s are $\T^d$-covariant, thereby ensuring that the resulting effective Lagrangians are $\T^d$-invariant. 
\end{enumerate}

\subsection{Perturbation of asymptotically free theories}

Suppose that $\mu=\set{\mu_n}$ is a resolution process with eventually independent increments, by which we mean that there exists some $n_0$ such that
\[
\langle f_ng_{n,m}\rangle = \langle f_n\rangle \langle g_{n,m}\rangle 
\]
for all  $m>n\geq n_0$\lo, $f_n\in C(X_n)$ and $g_{n,m}\in C(X_{n,m})$. Then, the densities 
\[
\tilde{\mc L}_n\bigl(x_n(s)\bigr) = x_n(s)^4 - 6\sigma_2(n)x_n(s)^2 + 6\sigma_2(n)^2-\sigma_4(n)
\]
are compatible for $n\geq n_0$\lo,
and one can define for $n<n_0$
\[
\mc L_n(x_n;s) = \EE{}{\tilde{\mc L}_{n_0}\bigl(x_{n_0}(s)\bigr)}{x_n}.
\]
Thus, one might expect that for cylinder measures satisfying a suitable weakened form of the eventually independent increment condition, the limit
\begin{equation} \label{failed renormalized lagrangian}
L_n(x_n) = \lim_{m\rightarrow\infty} \EE{}{\int_{\T^d}\tilde{\mc L}_m\bigl(x_m(s)\bigr)\D s}{x_n}
\end{equation}
should exist, but this is just not the case: an explicit calculation shows the presence of a divergent term as soon as $\int_{\T^d} x_{n,m}(s)^2\D s$ is not eventually independent of $x_n$\lo. We emphasize, however, that this does not rule out the existence of a compatible family $\mc L_n(x_n;s)$ such that
\begin{equation} \label{phi^4 condition}
\lim_{n\rightarrow\infty} \Bigl\{ \mc L_n(x_n;s) - \tilde{\mc L}_n\bigl(x_n(s)\bigr) \Bigr\} = 0.
\end{equation}
The problem upon trying to find it is that equations~\eqref{cubic}-\eqref{constant} allow, given  $\alpha_p(n)$, for a determination of  $\alpha_p(n+1)$ only up to terms with $\EE{}{\cdot}{x_n} = 0$. We produce next an explicit choice formally making condition~\eqref{phi^4 condition} hold, leaving the problem of uniqueness for further investigations. Our solution involves several limits clearly existing under conditions which will be referred to as \emph{asymptotic freedom} in an admittedly abusive way, because we do not  want to pick and stick to a precise definition yet. For the time being, having the theory of weakly dependent processes in mind (see~\cite{doukhan1994mixing}, for instance), we just note that the right notion will be related to the existence of a double sequence $\set{C_{nm}}$, satisfying some decay conditions, such that
\[
\bigl(\forall f_n\in C(X_n'),\ g_m\in C(X_m')\bigr)\quad \bigl\lvert \langle f_ng_m\rangle - \langle f_n\rangle\langle g_m\rangle \bigr\rvert \leq C_{nm} \norm{f_n}
\norm{g_m}
\]
in some suitable norm. As our calculations below suggest, for the $\phi^4$ field correlation decay conditions such as $\sum_{n,m} C_{nm}\sigma_4'(m) <\infty$ seem to be good candidates.

\begin{thm}
Let $\mu$ be a $\T^d$-invariant, even, 
asymptotically free process. We assume that our  manipulations can be justified if the correlation decay is fast enough. Define $\alpha_4(n)=1$ and
\begin{align*}
\alpha_3(n) &= 4\sum_{i=n+1}^\infty \EE{}{ x_i'(s) }{x_n} \\
\alpha_2(n) &= 3\sum_{i=n+1}^\infty\EE{}{\alpha_3(i)x_i'(s)}{x_n} + \lim_{m\rightarrow\infty}\biggl\{ 6\sum_{i=n+1}^m \EE{}{ x_i'(s)^2 }{x_n} - 6\sigma_2(m) \biggr\} \\
\alpha_1(n) &= 2\sum_{i=n+1}^\infty \EE{}{\alpha_2(i)x_i'(s)}{x_n} + 3\sum_{i=n+1}^\infty \EE{}{\alpha_3(i)x_i'(s)^2}{x_n} + 4\sum_{i=n+1}^\infty \EE{}{x_i'(s)^4}{x_n} \\
\alpha_0(n) &= \sum_{i=n+1}^\infty \EE{}{\alpha_1(i)x_i'(s)}{x_n} + \sum_{i=n+1}^\infty \EE{}{\alpha_3(i)x_i'(s)^3}{x_n} \\ &\quad + \lim_{m\rightarrow\infty}\biggl\{ \sum_{i=n+1}^m\EE{}{\alpha_2(i)x_i'(s)^2}{x_n} + \sum_{i=n+1}^m\EE{}{x_i'(s)^4}{x_n} + 6\sigma_2(m)^2 - \sigma_4(m) \biggr\} 
\end{align*}
Then, $\mc L_n(x_n;s) = \sum_{p=0}^4 \alpha_p(n;x_n;s)x_n(s)^p$ gives a perturbation of $\mu$ such that, formally,
\[
\lim_{m\rightarrow\infty}\Bigl\{ \mc L_m(x_m;s) - \tilde{\mc L}_m\bigl(x_m(s)\bigr) \Bigr\} = 0.
\]
\end{thm}
\begin{proof}
Let us check compatibility. We have that
\begin{align*}
&\EE{}{\alpha_3(n+1)+x_{n+1}'(s)}{x_n} \\ &\qquad = \sum_{i=n+2}^\infty\EE{}{x_i'(s)}{x_n} + \EE{}{x_{n+1}'(s)}{x_n} \\ &\qquad= \alpha_3(n), \\
&\EE{}{\alpha_2(n+1) + 3\alpha_3(n+1)x_{n+1}'(s) + 6x_{n+1}'(s)^2}{x_n} \\
&\qquad = 3\sum_{i=n+2}^\infty\EE{}{\alpha_3(i)x_i'(s)}{x_n} + 3\EE{}{\alpha_3(n+1)x_{n+1}'(s)}{x_n} \\ &\qquad\qquad+ \lim_{m\rightarrow\infty}\biggl\{ 6\sum_{i=n+2}^m\EE{}{x_i'(s)^2}{x_n} - 6\sigma_2(m) + 6\EE{}{x_{n+1}'(s)^2}{x_n} \biggr\}  \\
&\qquad = \alpha_2(n),
\end{align*}
and similar calculations hold for $\alpha_1$ and $\alpha_0$. 

Before proceeding we introduce some convenient notation. We will write
\[
\sigma_p(n,m) = \Langle x_{n,m}(s)^p\Rangle,\quad \sigma_p'(n) = \Langle x_n'(s)^p\Rangle.
\]
Observe that, thanks to the $\T^d$-invariance of $\mu$,
\[ 
 \Langle \sum_{i=n+1}^m\sum_{j=n+1}^m x_i'(s)x_j'(s) \Rangle 
	= \sum_{i=n+1}^m\Langle x_i'(s)^2\Rangle + \sum_{i\neq j}\Langle\strokedint_{\T^d} x_i(s)x_j(s)\D s\Rangle,
\]
and therefore $\sigma_2(n,m) = \sum_{i=n+1}^m \sigma_2'(i)$.

Now, it is clear that both $\alpha_3(n)$ and $\alpha_1(n)$ vanish as $n\rightarrow\infty$. The same happens to the first term of $\alpha_2(n)$ and the first and second terms of $\alpha_4(n)$. Also,
\[ 
\lim_{m\rightarrow\infty}\sum_{i=n+1}^m\biggl\{ \EE{}{x_i'(s)^2}{x_n} - \sigma_2(m) \biggr\} = -\sigma_2(n) +  \sum_{i=n+1}^\infty \EE{}{x_i'(s)^2-\sigma_2'(i)}{x_n},
\] 
showing that $\lim_{n\rightarrow\infty}\bigl\{ \alpha_2(n) + 6\sigma_2(n) \bigr\} = 0$. As for the third term in $\alpha_0(n)$, note first that $\sum_{i=n+1}^m\EE{}{\alpha_2(i)x_i'(s)^2}{x_n}$ has the following subterms:
\begin{align*}
&3\sum_{i=n+1}^m\sum_{j=i+1}^\infty \EE{}{\alpha_3(j)x_j'(s)x_i'(s)^2}{x_n} \\
&6\sum_{i=n+1}^m\sum_{j=i+1}^\infty \EE{}{\left(x_j'(s)^2-\sigma_2'(j)\right)x_i'(s)^2}{x_n} \\
&-6\sum_{i=n+1}^m \sigma_2(i)\EE{}{x_i'(s)^2}{x_n} 
\end{align*}
It will be convenient to rewrite the third one as
\[
-6\sum_{i=n+1}^m \sigma_2(n)\EE{}{x_i'(s)^2}{x_n} - 6\sum_{i=n+1}^m\sum_{j=n+1}^i \sigma_2'(j)\EE{}{x_i'(s)^2}{x_n}.
\]
On the other hand,
letting $u_n=x_n(s)$ and $u_i=x_i'(s)$ for $i=n+1,\dots,m$,
by parity one has
\begin{align*}
\sigma_4(m) &= 
\Langle \Bigl(\sum_i u_i\Bigr)^4 \Rangle = \sum_{i,j,k,l} \Langle u_iu_ju_ku_l\Rangle 
	= \Bigl\langle \sum_i u_i^4\Bigr\rangle + 3\Bigl\langle \sum_{i\neq j} u_i^2u_j^2 \Bigr\rangle \\
	&= \sigma_4(n) + \sum_{i=n+1}^m\sigma_4'(i) + 6\sum_{i=n+1}^m\Langle x_n(s)^2x_i'(s)^2\Rangle + 6\sum_{i=n+1}^m\sum_{j=i+1}^m \Langle x_i'(s)^2x_j'(s)^2\Rangle.
\end{align*}
Thus,
\( 
\sum_{i=n+1}^m\EE{}{\alpha_2(i)x_i'(s)^2}{x_n} + \sum_{i=n+1}^m\EE{}{x_i'(s)^4}{x_n} + 6\sigma_2(m)^2 - \sigma_4(m) 
\) 
can be written as the sum of the following terms:
\begin{align*}
&6\sigma_2(n)^2 - \sigma_4(n) \\
&\sum_{i=n+1}^m\Bigl(\EE{}{x_i'(s)^4}{x_n} - \sigma_4'(i)\Bigr) \\
&6\sum_{i=n+1}^m \Bigl( 2\sigma_2(n)\sigma_2'(i) - \sigma_2(n)\EE{}{x_i'(s)^2}{x_n} - \Langle x_n(s)^2x_i'(s)^2\Rangle \Bigr) \\
&6\sigma_2(n,m)^2 - 6\sum_{i=n+1}^m\sum_{j=n+1}^i \sigma_2'(j)\EE{}{x_i'(s)^2}{x_n} - 6\sum_{i=n+1}^m\sum_{j=i+1}^m \Langle x_i'(s)^2x_j'(s)^2\Rangle  \\
&3\sum_{i=n+1}^m\sum_{j=i+1}^\infty \EE{}{\alpha_3(j)x_j'(s)x_i(s)^2}{x_n}  \\
&6\sum_{i=n+1}^m\sum_{j=i+1}^\infty \EE{}{\left(x_j'(s)^2-\sigma_2'(j)\right)x_i'(s)^2}{x_n} 
\end{align*}
This shows, assuming that a fast correlation decay will suffice to make this expressions converge as $m\rightarrow\infty$, that 
\[
\lim_{n\rightarrow\infty}\bigl\{\alpha_0(n)-6\sigma_2(n)^2+\sigma_4(n)\bigr\}=0. \qedhere
\]
\end{proof}


\section{Closing remarks}

We now assess our work by comparing it with the standard approaches to the construction of quantum fields, and in doing so we assume that our program can be brought to completion. Along the way we will also
comment on some of its implicances.

We start by considering the rigorous construction of the $\phi^4$ field in $d=2$ pioneered by Glimm and Jaffe~\cite{glimm1968lambda}, see also~\cite{simon2015p}. Regarding the ultraviolet problem, which is the one addressed here, their method is based on making sense of
\begin{equation} \label{wick ordered field}
\lim_{n\rightarrow\infty} \mc L_n\bigl(x_n(s)\bigr)
\end{equation}
as either an operator acting on Fock space or a well-defined function of the Gaussian field. Succeeding in this enables one to construct either the Hamiltonian or the Euclidean (bona-fide) measure of the theory without much trouble. Our approach proceeds more slowly, by first constructing a cylinder measure. From this point of view no limit needs to be taken, because the densities $\set{\mc L_n}$ already define a perturbation of the free cylinder measure. However, this perturbation has to be extended to more general cylinder measures, and a related limit, namely
\(
\lim_{m\rightarrow\infty} \EE{}{\mc L_m\bigl(x_m(s)\bigr)}{x_n},
\)
might seem like a good candidate---except for the fact that it does not exist. This suggests that trying to construct renormalized powers of the field as in~\eqref{wick ordered field} might not be suitable as a general renormalization strategy. We will come back to this later. 

When the same approach is used to construct the $\phi^4$ field in $d=3$~\cite{glimm1968boson3}, see also~\cite{rivasseau1991from}, some difficult problems are encountered. To start with, the limit~\eqref{wick ordered field} does not exist if $\mc L_n$ is the fourth-degree Wick polynomial: one has to add further conterterms which depend not only on the cutoff scale, but also on the coupling constant $\lambda$. This fact raises two philosophical questions that we are able to answer now:
\begin{enumerate}
	\item
Why all of a sudden the counterterms become non-linear in $\lambda$? Well, for us this comes as no surprise, because the dependence of the perturbation on the measure forces the aggregate effective Lagrangians to be non-linear functions of the coupling constant. 
	\item 
Why Wick ordering ceases to play the central role it did in $d=2$? Well, in fact this is not the case: it ceases to play an obvious role because higher order terms in the aggregate effective Lagrangian start affecting the construction of renormalized powers of the field, but the non-aggregate effective Lagrangian should still be built upon Wick ordering. This also helps clarifying and complementing the introductory remarks in~\cite[Chapter 8]{baez1992introduction}.
\end{enumerate}
All of this brings about a change of perspective regarding couplings: for us, they are \emph{constant,} the only cutoff-dependent object being the (non-ag\-gre\-gate) effective Lagrangian. 
But nowadays couplings are construed as running---a perspective that probably seemed natural in view of the apparent insufficiency of Wick ordering, together with the need for counterterms non-linear in the coupling constants---and with that comes the acceptance of the idea that, in general, there is a multiplicity of equally valid renormalization schemes, leading to different theories that are all entitled to be considered the quantum version of a unique classical Lagrangian. We believe that this is actually a mistake, but before elaborating on this it is pertinent to move on to the perturbative approach to renormalization. 

The traditional starting point in perturbative renormalization theory  is the equation
\begin{equation} \label{traditional approach}
\mu_{\lambda,P}(\D x_P) \approx \e^{-\int\lambda\tilde{\mc L}_P(x_P(s))\D s}\mu_{0,P}(\D x_P)
\end{equation}
which is used to construct an aggregate effective Lagrangian density $\tilde{\mc L}_P$\lo, generally belonging to the largest pertinent class of Lagrangians containing the classical one. This approach has been made fully rigorous~\cite{costello2011renormalization}, but given its perturbative nature one should not expect it to produce a construction of the theory. The result is a suitable family of aggregate effective Lagrangians that depends on the choice of renormalization scheme. Therefore, it becomes important that the relevant undetermined parameters constitute a finite dimensional manifold, leading to a classification of theories roughly into renormalizable (good) and non-renormalizable (bad) ones. Our approach also offers some insight in this regard: one traditionally has, as input data to produce an aggregate effective Lagrangian, just a tangent vector (the perturbation of the free measure given by Wick ordering the interaction) whereas a whole vector field is actually needed. 
Now, as we have seen, equations~\eqref{quartic}-\eqref{constant} allow for the determination of $\alpha_p(n+1)$ in terms of $\alpha_p(n)$ only up to terms with $\EE{}{\cdot}{x_n}=0$, and this is where a choice of renormalization scheme enters the scene. However, our approach enables us to explicitely identify the right choice: it is that for which the resulting non-aggregate effective Lagrangians are asymptotically given by the Wick ordered interaction. We emphasize, moreover, that we are able to formally solve for this renormalized Lagrangian regardless of any perturbative renormalizability criterion.
In particular, we can talk about \emph{the} renormalized $\phi^4$ Lagrangian in arbitrary dimension. 

We now make one last point concerning any approach to the construction of quantum fields. Recall the fact that the formal perturbation $\mc L_n(\mu;x_n;s)$ of an asymptotically free measure that we have found cannot be recovered as $\lim_{m\rightarrow\infty} \EE{}{\tilde{\mc L}_m\bigl(x_m(s)\bigr)}{x_n}$ with $\tilde{\mc L}_n$ the corresponding Appell polynomial. This makes it plausible that a local continuum theory exists without being the limit of any local lattice approximation. In particular, in the case of the $\phi^4$ field, this means that the well-known non-existence theorem in $d>4$~\cite{aizenman1981proof} might not imply the non-existence of the continuum theory. 
Of course, we should not expect the $\phi^4$ field to exist in $d>3$ because there it ceases to be asymptotically free, but at any rate a full non-existence proof is still lacking.

We finish with a remark on potential practical applications. Consider the differential equation~\eqref{measure evolution}, governing the change of the measure with the coupling constant. If one is dealing with an asymptotically free theory for which the relevant energy scale can be estimated a priori, then one can replace it by the approximation
\[
\frac{\D}{\D\lambda}\mu_n = -\left(\int_{\T^d}\mc L_n(\mu_n;s)\D s\right)\mu_n\lo,
\]
which is a non-linear differential equation that one might be able to solve numerically. This could be a practical way to compute the theoretical value of physical quantities in non-perturbative regimes---and, depending on the actual value of $n$, might even be easier to deal with than the Feynman diagram approach in perturbative calculations.

\bibliographystyle{amsplain}
\bibliography{/storage/emulated/0/eratosthenes/mathphys}

\end{document}